\documentclass[11pt]{amsart}
\usepackage{amsmath,amsthm,amssymb,verbatim}

\onecolumn

\usepackage{fullpage}
\usepackage{amsmath, amssymb, enumerate, color, graphicx, bm, url, mathbbol}
\usepackage{appendix}

\newtheorem{theorem}{Theorem}[section]
\newtheorem{proposition}[theorem]{Proposition}
\newtheorem{corollary}[theorem]{Corollary}
\newtheorem{lemma}[theorem]{Lemma}

\numberwithin{equation}{section}

\DeclareMathOperator*{\E}{\mathbb{E}}

\DeclareMathOperator*{\sign}{sign}

\def \R {\mathbb{R}}

\def \e {\varepsilon}
\def \eps {\varepsilon}

\def \k {\kappa}

\def \moo| {\langle}
\def \< {\langle }
\def \> {\rangle }
\def \^ {\widehat}

\newcommand{\norm}[1]{\left \|#1\right \|}
\newcommand{\zeronorm}[1]{\norm{#1}_0}
\newcommand{\onenorm}[1]{\norm{#1}_1}
\newcommand{\twonorm}[1]{\norm{#1}_2}
\newcommand{\inftynorm}[1]{\norm{#1}_\infty}
\newcommand{\opnorm}[1]{\norm{#1}}

\newcommand{\abs}[1]{\left | #1 \right |}

\def \tanh{\text{tanh}}

\newcommand{\vect}[1]{\bm{#1}}

\def \va {\vect{a}}

\def \vg {\vect{g}}

\def \vv {\vect{v}}

\def \vx {\vect{x}}
\def \vy {\vect{y}}
\def \vz {\vect{z}}

\newcommand{\ip}[1]{\langle#1\rangle}
\newcommand{\parenth}[1]{\left(#1\right)}

\begin{document}
\title{One-bit compressed sensing with non-Gaussian measurements}
\author{Albert Ai}
\author{Alex Lapanowski}
\author{Yaniv Plan}
\author{Roman Vershynin}

\address{Department of Mathematics,
  University of Michigan,
  530 Church St.,
  Ann Arbor, MI 48109, U.S.A.}

\email{\{aflapan, yplan, romanv\}@umich.edu}
\email{aai@princeton.edu}

\thanks{Y.P.~was supported by an NSF Postdoctoral Research Fellowship under award No. 1103909. 
R.V.~was supported by NSF grant 1001829. A.A. and A.L.~were REU students
 supported by the NSF grant 0918623.}
 
 \date{Submitted August 2012}

\begin{abstract}
In one-bit compressed sensing, previous results state that sparse signals may be 
robustly recovered when the measurements are taken using Gaussian random vectors.
In contrast to standard compressed sensing, these results are not extendable to 
natural non-Gaussian distributions without further assumptions, as can be demonstrated by simple counter-examples involving extremely sparse signals.  
We show that approximately sparse signals that are not extremely sparse
can be accurately reconstructed from single-bit measurements sampled according to a sub-gaussian distribution, and the reconstruction comes as the solution to a convex program.


\end{abstract}

\maketitle

\section{Introduction}

In the standard noiseless compressed sensing model,
one has access to linear 
measurements of the form
$$y_{i}=\ip{\va_{i}, \vx},\qquad i=1, 2, \dots, m$$
where $\va_{1},\dots, \va_{m}\in\R^{n}$ are known measurement vectors and $\vx \in \R^n$ is a sparse signal which one wishes to reconstruct (see, e.g., \cite{CSBook}). 
Let $\zeronorm{\vx}$ denote the number of nonzero entries in $\vx$.  Typical results state that when the measurement vectors are chosen randomly from a sub-gaussian distribution, and $\zeronorm{\vx} \leq s$, then $m = O(s \log(n/s))$ measurements are sufficient for robust recovery of the 
signal $\vx$ (see, \cite{CSBook}).	

In noiseless one-bit compressed sensing, the measurements are compressed to single bits, and thus they take the form 
\begin{equation}
\label{eq:one-bit}
y_{i}=\sign\parenth{\ip{\va_{i}, \vx}},\qquad i=1, 2, \dots, m. 
\end{equation}
Here, the sign function is defined by $\text{sign}(t)=1$ when $t\geq 0$ and $-1$ otherwise.  Clearly, the magnitude of $\vx$ is lost in these measurements and so the goal is to approximate the direction of $\vx$. 
Thus we may assume without loss of generality that $\vx \in S^{n-1}$.  

One-bit compressed sensing was introduced in \cite{Boufounos2008} to model extreme quantization in compressed sensing.
The webpage \verb=http://dsp.rice.edu/1bitCS/= details the recent literature that concerns
theoretical and algorithmic results on one-bit compressed sensing, as well as applications and 
extensions to quantization with more than two bits. 
Let us review the existing theoretical results on one-bit quantization.  

Suppose that the signal $\vx \in \R^n$ satisfies $\zeronorm{\vx} \leq s$. 
Gupta et al. \cite{Gupta2010} assume that the measurement vectors $\va_i$ are Gaussian and demonstrate that 
the support of $\vx$ can tractably be recovered from either 
1) $O(s \log n)$ nonadaptive measurements assuming a constant dynamic range of $\vx$ 
(i.e. the magnitude of all nonzero entries of $\vx$ is assumed to lie between two constants), or 
2) $O(s \log n)$ adaptive measurements.  
Jacques et al. \cite{Jacques2011} introduce a certain \textit{binary $\epsilon$-stable embedding property} which is a one-bit analogue to the \textit{restricted isometry property} of standard compressed sensing.  They demonstrate that Gaussian measurement ensembles satisfy this property with high probability (given enough measurements).  Assuming the binary $\epsilon$-stable embedding property holds, they show that any estimate of $\vx$ which is both $s$-sparse and approximately matches the data, will be accurate.  In particular, $O(s \log n)$ Gaussian measurements are sufficient to have a relative error bounded by any fixed constant.  These results are robust to noise. 

Plan and Vershynin \cite{pv-1-bit, pv2012} show that one may reconstruct a sparse signal $\vx$ from single-bit measurements by {\em convex programming}, for which tractable solvers exist. 
\cite{pv-1-bit} considers the noiseless case and \cite{pv2012} considers the noisy case 
(and also sparse logistic regression). 
In \cite{pv2012} and the present paper, the model for the signal $\vx$ is allowed to be quite general, with sparsity as a special case. Indeed, suppose
$\vx$ belongs to some known set $K$, which is meant to encode the {\em model of the signal structure}.
For example, in order to encode sparsity, one could let $K$ be the set
$$
S_{n,s} := \{\vx \in \R^n: \zeronorm{\vx} \leq s, \twonorm{\vx} \leq 1\}.
$$
The recovery is achieved in \cite{pv2012} by solving the optimization problem 
\begin{equation}			\label{optimization}
\max \sum_{i=1}^{m}y_i \ip{ \va_{i}, \vx'} \qquad \text{subject to} \qquad \vx' \in K.
\end{equation}
If $K$ is a convex set then \eqref{optimization} is a convex optimization problem, 
so it can be solved by a variety of convex optimization solvers.

However, the reader may note that the set of sparse vectors $S_{n,s}$ is extremely non-convex.
To overcome this, it was proposed in \cite{pv2012} to take $K$ to be an approximate convex 
relaxation of $S_{n,s}$ (see \cite[Lemma 3.1]{pv-1-bit}), namely
\begin{equation}				\label{eq: Kns}
K = K_{n,s} := \{\vx \in \R^n: \twonorm{\vx} \leq 1, \onenorm{\vx} \leq \sqrt{s}\}.
\end{equation}
It was shown in \cite{pv2012} that $m = O(s \log(n/s))$ Gaussian measurements are sufficient
to accurately recover $\vx$ by solving the convex optimization problem \eqref{optimization}.

\medskip

A natural question is whether reconstruction of $\vx$ from one-bit measurements is still feasible when measurements are taken using random vectors with {\em non-Gaussian} coordinates. 
A simple counterexample shows that this is not generally possible even when the coordinates are sub-gaussian. Suppose that 
all coordinates of $\va_{i}$ are in $\{-1,1\}$. For example, one may let 
the coordinates be independent symmetric Bernoulli random variables.
Then the vectors
$$
\vx = (1, \frac{1}{2}, 0, \ldots,0) \quad \text{and} \quad \vx' = (1, -\frac{1}{2}, 0, \ldots,0)
$$ 
clearly satisfy
$\sign\parenth{\ip{\va_{i}, \vx}} = \sign\parenth{\ip{\va_{i}, \vx'}} $. 
This shows that one can not distinguish the two very different signals $\vx$ and $\vx'$ 
by such measurements,\footnote{One can normalize the 
signals $\vx$ and $\vx'$ to lie on $S^{n-1}$, and the same phenomenon clearly persists.} 
even if infinitely many measurements are taken.




One may ask whether this counterexample has typical or worst-case behavior. 
In this paper, we demonstrate that the latter is the case---{\em a difficulty can only arise 
for extremely sparse signals}. Namely, we show that under the assumption
\begin{equation}				\label{eq: sup norm small}
\|\vx\|_\infty \ll \|\vx\|_2 = 1,
\end{equation}
an approximate recovery of $\vx$ is still possible with general sub-gaussian measurements, 
and it is achieved by the convex program \eqref{optimization}. 
Furthermore, we prove that for the distributions that are near Gaussian (in total variation), 
an approximate recovery of $\vx$ is possible even without
the assumption \eqref{eq: sup norm small}.






\subsection{Main Results}		\label{sec:main results}

We shall assume that the signal set $K$ lies in the unit Euclidean ball in $\R^n$, 
which we shall denote $B_2^n$. 
The quality of recovery of a signal $\vx \in K$ will depend on $K$ through a 
single geometric parameter -- the {\em Gaussian mean width} of $K$. It is defined as 
$$
w(K) = \E \sup_{\vx \in K - K} \ip{\vg, \vx},
$$
where $\vg$ denotes a standard Gaussian random vector in $\R^n$, i.e. a vector
with independent $N(0,1)$ random coordinates.
The reader may refer to \cite[Section 2]{pv2012} for a brief overview of the 
properties of mean width. 

\medskip

The main purpose of this paper is to allow the measurement vectors $\va_i$ to 
have general {\em sub-gaussian} (rather than Gaussian) independent coordinates. 
Recall that a random variable $a$ is sub-gaussian if its distribution is dominated 
by a centered normal distribution. This property can be expressed in several equivalent ways, 
see \cite[Section 5.2.3]{Vershynin2012}. One convenient way to define a sub-gaussian 
random variable is to require that its moments
be bounded by the corresponding moments of $N(0,1)$, so that 
$(\E |a|^p)^{1/p} = O(\sqrt{p})$ as $p \to \infty$. 
Formally, $a$ is called sub-gaussian if
\begin{equation}  \label{eq:subgaussian}
\k := \sup_{p \ge 1} p^{-1/2} (\E |a|^p)^{1/p} < \infty.
\end{equation}
The quantity $\k$ is called the {\em sub-gaussian norm} of $a$.
The class of sub-gaussian random variables includes in particular normal, Bernoulli and all bounded random variables.

\medskip

Our main result is a generalization of \cite[Theorem 1.1]{pv2012}, which 
states that when the measurement vectors $\va_{1}, \dots, \va_{m}$ are Gaussian, then 
\[
\|\vx-\widehat{\vx}\|_{2}^{2}\lesssim \frac{w(K)}{\sqrt{m}}
\]
with high probability.
Our generalization allows $\va_{i}$ to have coordinates with sub-gaussian distributions. 
The only important difference is that the error now has an additive dependence on $\|\vx\|_\infty$.
This serves to exclude extremely sparse signals, which can destroy recovery,
according to the example we discussed above. 

We will consider the noisy measurement model in which each 1-bit measurement is flipped with small probability
\begin{equation}
\label{eq:bit-flip model}
y_i = \eps_i \sign( \< \va_i, \vx \> ).
\end{equation}
Above, $\{\eps_i\}$ is a collection of i.i.d.~Bernoulli random variables satisfying $P(\eps_i = 1) = 1-p$ and $P(\eps_i = -1) = p$.
\begin{theorem}[Estimating a signal with random bit flips]\label{thm2}
  Let $a \in \R$ be a symmetric, sub-gaussian, and unit variance random variable with 
  $\k$ as in \eqref{eq:subgaussian}.  Let $\va_1, \dots, \va_m$ be 
  independent random vectors in $\R^n$ whose coordinates 
  are i.i.d.~copies of $a$. Consider a signal set $K\subseteq B_{2}^{n}$, and fix 
  $\vx \in K$ satisfying $\|\vx\|_2 = 1$. Let $\vy$ follow the 1-bit measurement model 
  of Equation \eqref{eq:bit-flip model}.  Then for each $\beta > 0$, with probability at least 
  $
  1 - 4e^{-\beta^2},
  $
  the solution 
  $\widehat{\vx}$ to the optimization problem~(\ref{optimization}) satisfies
  \begin{equation}
  \label{eq:no noise thm bound}
  \|\vx - \widehat{\vx}\|_2^2 \leq C\left( \k^3\|\vx\|_\infty^{1/2} + \frac{\k}{\sqrt{m}(2 - p)}(w(K) + \beta)\right).
  \end{equation}
\end{theorem}
In this theorem and later, $C$ and $c$ denote positive absolute constants, which can be different
from line to line.

A proof of Theorem~\ref{thm2} is given in Section~\ref{thm2p}.

\medskip

This theorem can be easily specialized to sparse (and approximately sparse) signals.
To this end, we consider $K = K_{n,s}$ as in \eqref{eq: Kns}.  
A standard computation (see \cite[Equation 3.3]{pv2012}) shows that
\[w(K_{n,s}) \leq C\sqrt{s \log(2n/s)}.\]
Then the following corollary follows directly from Theorem \ref{thm2}.

\begin{corollary}[Estimating a sparse signal with random bit flips]			\label{cor:sparse}
Let $K = K_{n,s}$, $s \geq 1$, and let everything else be as in Theorem \ref{thm2}.  Then with probability at least
$
1 - 4\exp\left\{-2 s \log(2n/s)\right\} \geq 1 - \frac{1}{n^2},
$
the solution $\widehat{\vx}$ to the optimization problem~(\ref{optimization}) satisfies
\[
\|\vx - \widehat{\vx}\|_2^2 \leq C\left(\k^3 \|\vx\|_\infty^{1/2} + \k \sqrt{\frac{s \log(n/s)}{m(2-p)^2}}\right).
\]
\end{corollary}

In words, this result yields that if the signal is approximately $s$-sparse, but not extremely sparse so that
$\|\vx\|_\infty \ll \|\vx\|_2 = 1$, then with high probability $\vx$ can be accurately recovered 
from 
$$
m = O(s \log (n/s))
$$
general sub-gaussian measurements---provided that at most a constant fraction of bits are randomly flipped.  Interestingly, accurate reconstruction is still possible even if nearly half of the bits are flipped.

\medskip

We also establish a version of Theorem \ref{thm2} under a statistical model, which also corresponds to additive noise before quantization. 
We take the \textit{generalized linear model} in which each measurement is modeled by a random variable $y_i$ 
taking values in $\{-1,1\}$
such that 
\begin{equation}\label{noise}
\E(y_i | \va_i) = \theta(\ip{\va_i, \vx}), \qquad i=1, 2, \dots, m.
\end{equation}
Conditionally on $\{\va_i\}$, the measurements  $y_i$ are assumed independent. 
$\theta:\R^{d}\to [-1,1]$ is a measurable function, which may even be unknown or unspecified.  
We only assume that $\theta(t) \in C^3(\R)$, the first three derivatives being bounded by $\tau_1, \tau_2, \tau_3$ respectively, and that
\begin{equation}			\label{lambda}
\E\theta(g)g =: \lambda > 0
\end{equation}
where $g \sim N(0,1)$. To see why this is a natural assumption, notice that 
$\ip{\va_{i}, \vx}\sim\mathcal{N}(0,1)$ if $\va_{i}$ are standard Gaussian random vectors
and $\|x\|_{2}=1$; thus
\[
\E y_{i}\ip{\va_{i}, \vx}=\E\theta(g)g=\lambda. 
\]


For example, in sparse 
 logistic regression one would take 
\[
\theta(t)= \tanh (t/2),
\]
with bounds $\tau_1 = 0.5$, $\tau_2\approx 0.19$, $\tau_3\approx 0.083$ and $\lambda \approx 0.41$. 

To note another important example, observe that the setting of Theorem \ref{thm2} is described by choosing $\theta(t) = \sign(t)$ and disregarding the differentiability requirements. In this case, $\lambda = \E \theta(g)g=\E\abs{g}=\sqrt{2/\pi}$.

The following is a version of Theorem \ref{thm2} under this noisy or statistical model.

\begin{theorem}[Estimating a spread signal in the generalized linear model]\label{thm1}
  We remain in the setting of Theorem~\ref{thm2}, but with random measurements $y_i$ modeled as in   
  Equation (\ref{noise}). Then for each $\beta > 0$, with probability at least 
  $
  1 - 4e^{-\beta^2},
  $
  the solution $\widehat{\vx}$ to the optimization 
  problem~(\ref{optimization}) satisfies
  \begin{equation}\label{eq:random noise theorem}
  \|\vx - \widehat{\vx}\|_2^2 \leq C \left(\frac{\k^4}{\lambda}
  (\tau_2 + \tau_3)\|\vx\|_\infty  + \frac{\k}{\lambda\sqrt{m}}(w(K) + \beta) \right).
  \end{equation}
\end{theorem}

For Gaussian measurement vectors $\va_i$, a version of this theorem was proved in 
\cite{pv2012}.

The proof of Theorem~\ref{thm1} is provided in Section~\ref{thm1p}.

An interested reader may specialize this result to sparse signals $\vx$ 
as we did before, i.e. by taking $K = K_{n,s}$ and noting as in Corollary \ref{cor:sparse} 
that $w(K_{n,s}) \leq C \sqrt{\log(2n/s)}$.

\medskip

Our last result is about sub-Gaussian distributions, which nevertheless are close to Gaussian 
in total variation. For such measurements, it is reasonable to expect that the same conclusions 
as for Gaussian measurements hold, i.e.,~that the theorems above 
hold for all signals $\vx$ without any dependence on $\|\vx\|_\infty$.
We confirm that this is the case. Suppose that the coordinates of $\va_i$ are i.i.d. copies 
of a random variable $a$ that satisfies the total variation bound
$$
\|a - g\|_{TV} := \sup_{A}|P(a \in A) - P(g \in A)| \leq \eps
$$
where $g \sim N(0,1)$. 
In the case when $\theta(t) = \sign(t)$, one has 
\[
\|\vx -\widehat{\vx}\|_{2}^{2} \lesssim \eps^{1/8}+\frac{w(K)}{\sqrt{m}},
\]
and in the case when $\theta(t) \in C^2$ one has
\[
\|\vx-\widehat{\vx}\|_{2}^{2} \lesssim \eps^{1/2}+\frac{w(K)}{\sqrt{m}}. 
\]
Above, the $\lesssim$ notation hides dependence on $\kappa, \lambda, \tau_1, \tau_2$ and a numeric constant.
The precise results and their proofs are provided in the appendix as 
Theorems $\ref{thm3}$ and $\ref{thm4}$, respectively. 

\section*{Acknowledgements}
The authors would like to thank the anonymous reviewers for thoughtful comments which helped to simplify the presentation and tighten results.

\section{General proof structure}
In this section we give the general structure behind the proofs of our theorems.  We also give a general lemma which may be useful to other researchers who wish to develope theory for 1-bit compressed sensing under different measurement models.

It will be convenient to define the (rescaled) objective function for our convex program \eqref{optimization}:
\[
 f_{\vx} (\vx') := \frac{1}{m} \sum_{i = 1}^m y_i \ip{\va_i, \vx'}.
\]
Note that this is a random function whose distribution depends on the distribution of $\{\va_i\}$ and choice of $\theta$.  In order to demonstrate that $\hat{\vx}$ is a good approximation of $\vx$, we will need to control the expectation of $f$ and the variation of $f$ around its expectation.  It turns out that $f$ uniformly concentrates around its expectation value regardless of $\theta$, but with dependence on the sub-Gaussian norm of $\va_i$.

\begin{proposition}[Concentration]\label{concentration}
  For each $\beta > 0$,
  \[
  P\parenth{\sup_{\vz\in K - K} |f_{\vx}(\vz) - \E f_{\vx}(\vz)| \geq 
  C \k \frac{w(K) + \beta}{\sqrt{m}} } \leq 4 e^{-\beta^2}.
  \]
\end{proposition}
The proof of Proposition~\ref{concentration} is provided in Section~\ref{concentrationp}.

We turn to the expectation of $f_{\vx}(\vx')$.  In the special case of standard normal measurement vectors $\va_i$, it is not hard to show that
\[\E f_{\vx}(\vx') = \lambda \< \vx, \vx' \> .\]
(See \cite[Lemma 4.1]{pv2012}.)  If we allow sub-gaussian measurement vectors, the equality no longer holds, but under some conditions on $\vx$ and $\theta$, we still have
\[\E f_{\vx}(\vx') \approx \lambda \< \vx, \vx' \> .\]
We will prove the above approximate equality in subsequent sections, but for now we see how it implies accurate reconstruction of $\vx$.

\begin{lemma}
\label{lem:general structure}
Fix $\vx \in K$ with $\twonorm{\vx} = 1$.  Let $\theta: \R \rightarrow [-1, 1]$ be measurable and suppose that $\vy$ follows the generalized linear model \eqref{noise}.  Let $\alpha > 0$ and suppose that for any $\vx' \in K$,
\begin{equation}
\label{eq:bound expectation}
\abs{\E f_{\vx}(\vx') - \lambda \< \vx, \vx' \>} \leq \alpha.
\end{equation}
Then for all $\beta > 0$, with probability at last $1 - 4 \e^{-\beta^2}$ the solution $\hat{\vx}$ to \eqref{optimization} satisfies
\[\twonorm{\hat{\vx} - \vx}^2 \leq \frac{4\alpha}{\lambda} + C \kappa \frac{w(K) + \beta}{ \lambda \sqrt{m}}.\]
\end{lemma}

\begin{proof}
Fix $\vx' \in K$.  We will show that $f_{\vx}(\vx')$ can only be large if $\vx'$ is near to $\vx$, and then use this to show that the maximizer $\hat{\vx}$ must be accurate.  

Let $\vz = \vx' - \vx \in K - K$. 
We have,
\[-\E f_{\vx}(\vz) = \E f_{\vx}(\vx) - \E f_{\vx}(\vx') \geq \ip{ \lambda \vx, \vx }  
- \ip{ \lambda \vx, \vx' } - 2\alpha \geq \frac{\lambda}{2} \twonorm{ \vx - \vx'}^2 - 2 \alpha.\]
The first inequality follows from Equation \eqref{eq:bound expectation} and the second follows from $\twonorm{\vx} = 1$.

Further, by Proposition~\ref{concentration}, we have a lower bound of
$
1- 4e^{-\beta^2}
$
on the event
$$\sup_{\vz\in K - K} |f_{\vx}(\vz) - \E f_{\vx}(\vz)| \le C \k \frac{w(K) + \beta}{\sqrt{m}}.$$
In this event, note that
\[
 f_{\vx}(\vz) \leq \E f_{\vx}(\vz) + C \k \frac{w(K) + \beta}{\sqrt{m}}
 \leq 2 \alpha - \frac{\lambda}{2} \|\vx - \vx'\|_2^2 + C \k \frac{w(K) + \beta}{\sqrt{m}}.\]
This holds uniformly for all $\vx' \in K$.  Pick $\vx' = \widehat{\vx}$.  Since $\hat{\vx}$ maximizes $f_{\vx}$ we have $f_{\vx}(\vz) = f_{\vx}(\widehat{\vx}) - f_{\vx}(\vx) \geq 0$. Thus the right-hand side of the above inequality is bounded below by 0.   Rearranging completes the proof of the lemma.
\end{proof}

\section{Proof of Theorem~\ref{thm1}}\label{thm1p}

We only need to bound $\alpha$ in Equation \eqref{eq:bound expectation}.  
For convenience, let us denote $y := y_1$ and $\va := \va_1$.
Recalling \eqref{noise}, we observe the following equivalences:
\begin{equation}							\label{expectation fx}
\E f_{\vx} (\vx') 
= \frac{1}{m} \sum_{i = 1}^m \E y_i \ip{ \va_i, \vx'} 
= \E y \ip{ \va, \vx'} 
= \E(\E y\ip{ \va, \vx' } | \va) 
= \E\theta(\ip{ \va, \vx }) \ip{ \va, \vx' }.
\end{equation}
We also note that for a standard normal vector $\vg$, $\E\theta(\ip{ \vg, \vx }) \ip{ \vg, \vx' } = \lambda \< \vx, \vx' \> $, which satisfies Equation \eqref{eq:bound expectation} with $\alpha = 0$ (see \cite[Lemma 4.1]{pv2012}).  Thus, we need to show that the expectation in the sub-gaussian case nearly matches the Gaussian case.  
Such a comparison is a bi-variate version of Berry-Esseen central limit theorem for the function  
$\theta(\ip{ \va, \vx }) \ip{ \va, \vx' }$.  

\begin{lemma}[Berry-Esseen type central limit theorem] \label{fourth}
  Consider $\vx, \vz \in B_2^n$.  Let $\va$ be a random vector with i.i.d.~mean-zero, variance-one, sub-gaussian entries whose sub-gaussian norm is bounded by $\kappa$.  Let $\vg$ be a vector with independent standard normal entries. Let $\theta :\R \rightarrow \R$ be a measurable function satisfying $\norm{\theta''}_\infty \leq \tau_2$ and $\norm{\theta'''}_\infty \leq \tau_3$.  Then 
  \begin{equation}
  \label{eq:berry esseen}
|\E \theta(\ip{ \va, \vx })\ip{\va, \vz } - \E \theta(\ip{ \vg, \vx })\ip{ \vg, \vz } |
\leq  C(\tau_2 + \tau_3)\E a^4 \| \vx \|_\infty,
\end{equation}
\end{lemma}
The proof is based on a Lindeberg replacement argument in two variables; it is provided in the appendix. 
Note that the quality of approximation in this theorem is the same for all $\vz \in B_2^n$; this will be crucial for our argument. 

We apply the lemma to prove Theorem~\ref{thm1}.
\begin{proof}[Proof of Theorem~\ref{thm1}.] 
We set $\alpha$ to be the right-hand side of Equation \eqref{eq:berry esseen}.  Further, by definition, $\E a^4 \leq 16 \kappa^4$.  Thus, we may apply Lemma~\ref{lem:general structure} with 
\[\alpha \leq C \kappa^4 (\tau_2 + \tau_3) \| \vx \|_\infty\]
to complete the proof. \qedhere

\end{proof}

\section{Proof of Theorem~\ref{thm2}}		\label{thm2p}
For simplicity, we first assume the noiseless model in \eqref{eq:one-bit}; we will then describe a minor adjustment to the proof to generalize to the random bit flip model in \eqref{eq:bit-flip model}.  

The essential difference from the proof of Theorem \ref{thm1} is that $\theta(t) = \sign(t)$ is not differentiable.   One approach would be to approximate $\theta$ by a smooth function and apply the Berry-Esseen type central limit theorem (Lemma~\ref{fourth}).  However, we achieve a tighter bound with a different approach.
Once again, we only need to bound $\alpha$ in Equation \eqref{eq:bound expectation}.  This bound is contained in the following proposition.  In this section, $\lambda =\E \abs{g} = \sqrt{2/\pi}$.

\begin{proposition}[Expectation]\label{expectation2}
Consider $\vx \in S^{n-1}, \vx' \in B_2^n$. If $\|\vx\|_\infty \leq c/\E|a|^3,$ then
\begin{equation}
\label{eq:expectation2 bound}
|\E f_{\vx}(\vx') - \lambda \ip{ \vx, \vx' }|\leq C\E|a|^3\|\vx\|_\infty^{1/2}.
\end{equation}
\end{proposition}

We pause to prove Theorem~\ref{thm2} based on the proposition.  
\begin{proof}[Proof of Theorem~\ref{thm2}]
By definition of $\kappa$, we have $\E\abs{a}^3 \leq 3^{3/2} \kappa^3$.  Now suppose that $\inftynorm{\vx} \geq c/\E \abs{a}^3$.  Thus, $\inftynorm{\vx} \geq c'/\kappa^3$.  Then the right-hand side of Equation \eqref{eq:no noise thm bound} in Theorem \ref{thm2} is lower bounded by
\[C\cdot c' \kappa^{3/2} \geq C \cdot c'.\]
Take $C \geq 4/c'$ in which case the theorem trivially holds since $\hat{\vx}, \vx \in B_2^n$.  

On the other hand, if $\inftynorm{\vx} \geq c/\E \abs{a}^3$ we can apply Proposition~\ref{expectation2}.  
We apply the proposition to bound $\alpha$ in Lemma~\ref{lem:general structure}.  This gives 
\[\alpha = C\kappa^3\|\vx\|_\infty^{1/2}\]
and completes the proof of the theorem in the noiseless case.

In the random bit flip model of Equation \eqref{eq:bit-flip model}, $\E f_{\vx}(\vx')$ is scaled by a factor of $\E \eps_i = 2 - p$.  This has the effective of scaling $\lambda$ and the right-hand side of Equation \eqref{eq:expectation2 bound} by a factor of $2-p$.  Thus, we complete the proof by applying Lemma \ref{lem:general structure} with
\[\alpha = C (2 - p) \kappa^3\|\vx\|_\infty^{1/2} \quad \text{and} \quad \lambda = \sqrt{2/\pi}(2 - p).\]
\end{proof}

In order to prove Proposition~\ref{expectation2}, we will need to use two different known one-dimensional 
Berry-Esseen results (see \cite[Theorems 2.1.24 and 2.1.30]{Strk}), stated below in simplified form for the convenience of the reader:

\begin{theorem}[One-dimensional Berry-Esseen central limit theorem]
Let $\vz$ be a random vector with $n$ independent, mean-zero, entries satisfying $\E \twonorm{\vz}^2 = 1$.  Set
\[S_n = \sum_{i = 1}^n z_i \quad \text{and} \quad \beta^3 := \E \norm{\vz}^3_3 := \sum_{i =1}^n \E \abs{z_i}^3\]
and let $g$ be a standard normal random variable.
We have
\begin{equation}
\label{eq:strk one norm}
\int_{-\infty}^\infty \abs{P( S_n \leq t) - P(g \leq t)} dt \leq 9 \beta^3
\end{equation}
and
\begin{equation}
\label{eq:strk inf norm}
\sup_{t} \abs{P(S_n \leq t) - P(g \leq t)} \leq 10 \beta^3.
\end{equation}
\end{theorem}

We now prove Proposition~\ref{expectation2} using a simple geometric argument.  Define the vector
\[\vv_x:= \E \sign( \< \va, \vx \> ) \va\]
and note that the conclusion of Proposition~\ref{expectation2} states that
\[\vv_x \approx \lambda \vx.\]
In order to prove that the above approximate equality holds, we will derive it from two \textit{scalar} inequalities:
\begin{equation}
\label{eq:approx ident}
\< \vv_x, \vx \> \approx \lambda \quad \text{and} \quad \twonorm{\vv_x} \lesssim \lambda.
\end{equation}
(Indeed, the first of these approximate identities states that $\vv_x$ is near a hyperplane with normal $\vx$, and the second one states that $\vv_x$ is nearly in the ball which intersects that hyperplane at the point $\lambda \vx$.)  
Thus we reduce the problem to proving \eqref{eq:approx ident}.

\begin{lemma}\label{first2}
$\abs{\ip{ \vv_x, \vx } - \lambda} \leq C\E|a|^3\|\vx\|_3^3 \leq C \E \abs{a}^3 \inftynorm{\vx}.$
\end{lemma}

\begin{proof}
Recall that by definition of $\vv_{x}$,
$$\ip{ \vv_x, \vx } = \E\sign(\ip{ \va, \vx }) \ip{ \va, \vx }= \E \abs{\< \va, \vx\> }.$$
Note that $\lambda = \sqrt{2/\pi} = \E \abs{g}$ and thus, to prove the lemma, we wish to bound the difference $\big \vert \E \abs{\< \va , \vx \>} - \E \abs{g} \big \vert$.  We have
\[
\big \vert\E \abs{ \< \va, \vx \>} - \E \abs{g}\big \vert  = \abs{\int_0^\infty P( \abs{ \< \va, \vx \>} \geq t ) - P( \abs{g} \geq t) dt}= 2 \abs {\int_0^\infty P(  \< \va, \vx \> \geq t ) - P( g \geq t) dt}.
\]
To bound the right-hand side, we apply the Berry-Esseen result in Equation \eqref{eq:strk one norm} which bounds the above quantity by
$$ C\sum_{i=1}^{n} \E|x_i a_i|^3 = C\E|a|^3\|\vx\|_3^3.$$
\qedhere

\end{proof}

To bound $\twonorm{\vv_x}$ we will apply the Berry-Esseen Theorem with $\vz = \vv_x/\twonorm{\vv_x}$.  This will require first a rough two-sided bound on $\twonorm{\vv_x}$ and also an upper bound on $\inftynorm{\vv_x}$.  We establish these in the following two lemmas.

\begin{lemma}  \label{zeroth} Suppose that $\inftynorm{\vx} \leq c/\E \abs{a}^3$.  Then $\frac{1}{2} \leq \twonorm{\vv_x} \leq 1$.
\end{lemma}
\begin{proof}
For the lower bound, using Lemma~\ref{first2}, we have
\[\twonorm{\vv_x} = \twonorm{\vv_x} \twonorm{\vx} \geq \abs{\< \vv_x, \vx \> } \geq \lambda - C \E \abs{a}^3 \inftynorm{\vx}.\]
Since $\lambda = \sqrt{2/\pi}$, and $\inftynorm{\vx} \leq c/\E \abs{a}^3$, the right-hand side is greater than $1/2$, as long as we take $c \leq (\sqrt{2/\pi} - 1/2)/C$.

In the other direction, we have
\begin{equation}
\label{eq:zeroth}
\twonorm{\vv_x}^2 = \< \vv_x, \vv_x \> = \E \sign( \< \va, \vx \> ) \< \va, \vv_x \> \leq \E \abs{ \< \va, \vv_x \> } \leq (\E \< \va, \vv_x \> ^2)^{1/2} = \twonorm{\vv_x}.
\end{equation}
It follows that $\twonorm{\vv_x} \leq 1$.
\end{proof}

\begin{lemma}\label{second2}  Suppose that $\inftynorm{\vx} \leq c/ \E\abs{a}^3$.  Then,
$\|\vv_x\|_\infty \leq C\E|a|^3\|\vx\|_\infty.$
\end{lemma}

\begin{proof}

Establishing the notation $\ip{ \va, \vx } = \sum_{k = 1}^n a_k x_k$ where without loss of 
generality, $x_i \geq 0$, define for convenience $S = \sum_{k\neq i}^n a_k x_k$ and let $F_S$ be the cumulative distribution function of $S$. Consider an 
arbitrary constant $r$. 

\begin{align*}
\abs{\E \theta(S +  r x_i)\cdot r} 
  &= \abs{r \int_\R \sign(t + r x_i) dF_S(t)} \\
  &= |r| \abs{P(S \geq -r x_i) - P(S < -r x_i)} = |r|P(|S| \leq |r| x_i) \\
  &\leq \abs{r} P(\abs{g} \leq \abs{r} x_i) + \abs{r} \cdot \abs{P(\abs{g} \leq \abs{r} x_i) - P(|S| \leq |r| x_i)}.
\end{align*}
The second term in the last inequality may be bounded using the Berry-Esseen result in Equation \eqref{eq:strk inf norm}.  This gives

$$ \abs{\E \theta(S +  r x_i)\cdot r} \leq |r| \bigg\{\sqrt{\frac{2}{\pi}}|r| x_i + 2\left(\sum_{k \neq i} x_k^2\right)^{-3/2} \E|a|^3\sum_{k \neq i} |x_k|^3\bigg\}. $$

Note $\|\vx\|_3^3 \leq \|\vx\|_\infty\|\vx\|_2^2 = \|\vx\|_\infty \leq c/ \E\abs{a}^3 \leq 1/8$, where the last inequality follows since $\E\abs{a}^3 \geq (\E a^2)^{3/2} = 1$ and we take $c \leq 1/8$. Then $x_i^3 \leq 1/8$, $x_i^2 \leq 1/4$, so that $\sum_{k \neq i} x_i^2 \geq 3/4.$
Observing furthermore that $\|\vx\|_\infty \geq \sum_{k \neq i} x_k^2 \|\vx\|_\infty \geq \sum_{k \neq i} |x_k|^3$, we have the bound 
$$ \abs{\E \theta(S +  r x_i)\cdot r} \leq C\abs{r}^2 x_i + C \abs{r} \E|a|^3 \|\vx\|_\infty.$$

We may express a single coordinate of $v_x = \E \theta( \ip{ \va, \vx }) \cdot \va$ as $\E \theta( \ip{ \va, \vx }) \cdot a_i$. Then,
\begin{align*}
|\E \theta( \ip{ \va, \vx }) \cdot a_i| 
  &\leq \int_\R \abs{\E \theta(S + t x_i) \cdot t}  dF_{a_i}(t) \\
  &\leq \int_\R \big( Ct^2 x_i + C|t|\E|a|^3 \|\vx\|_\infty \big) dF_{a_i}(t) \\
  &= C x_i \E a_i^2 + C \E|a|^3 \|\vx\|_\infty \E|a_i| \\
  &\leq C x_i + C \E|a|^3 \|\vx\|_\infty.
\end{align*}

Observing that $\E|a|^3 \geq \E a^2 = 1$ completes the proof of the lemma. \qedhere

\end{proof}

We now prove Proposition~\ref{expectation2}.  
\begin{proof}[Proof of Proposition~\ref{expectation2}]
Define $\vz = \vv_x/\|\vv_x\|_2$ and note that $ \|\vz\|_\infty = \inftynorm{\vv_x}/\twonorm{\vv_x}$.   Applying Lemma~\ref{zeroth} and Lemma~\ref{second2} yields $ \|\vz\|_\infty \leq CE|a|^3\|\vx\|_\infty$. Hence,
\begin{align*}
\|\vv_x\|_2 &= \ip{\vv_x, \vv_x/\|\vv_x\|_2} = \E \sign(\ip{ \va, \vx}) \ip{ \va, \vz}\\
&\leq \E \sign(\ip{ \va, \vz}) \ip{ \va, \vz} = \ip{\vv_z, \vz}\\
&\leq \lambda + C\E|a|^3\|\vz\|_\infty \leq \lambda + C(\E|a|^3)^2\|\vx\|_\infty.
\end{align*}
In the last line, we used Lemma~\ref{first2}.  Together with Lemma \ref{first2} we have now verified both geometric constraints in \eqref{eq:approx ident}.

Combining results, we have
\begin{align*}
|\ip{ \vv_x, \vx' } - \ip{ \lambda \vx, \vx' }|^2 &\leq \|\vv_x\|_{2}^2 - \lambda^2 + 2\lambda(\lambda - \ip{ \vv_x, \vx })\\
&= (\|\vv_x\|_{2} + \lambda)(\|\vv_x\|_{2} - \lambda) + 2\lambda(\lambda - \ip{ \vv_x, \vx })\\
&\leq C\left((\E|a|^3)^2\|\vx\|_\infty + \E|a|^3\|\vx\|_3^3 \right).
\end{align*}
Recall that $\|\vx\|_\infty \geq \|\vx\|_3^3$ and thus the first term is dominant.  We may collect terms to conclude
$$|\E f_{\vx}(\vx') - \ip{ \lambda \vx, \vx' }|\leq C \E|a|^3\|\vx\|_\infty^{1/2}.$$
This completes the proof of Proposition~\ref{expectation2}. \qedhere
\end{proof}

\subsection{Concentration: Proof of Proposition~\ref{concentration}}\label{concentrationp}

We need to control the random variable
$$
Z := \sup_{\vz\in K - K} |f_{\vx}(\vz) - \E f_{\vx}(\vz)|.
$$
This will be done using techniques from probability in Banach spaces, following the argument in \cite[Proposition 4.2]{pv2012}. 
The symmetrization lemma below allows us to essentially replace $Z$ by the random variable
$$
Z' := \sup_{\vz\in K - K} \frac{1}{m} \abs{ \sum_{i = 1}^m \eps_i y_i\ip{ \va_i, \vz }}.
$$
where $\eps_i$ denote independent symmetric Bernoulli random variables.

\begin{lemma}[Symmetrization]\label{symmetrization}
  We have
  \begin{equation}\label{sym1}
    \E Z \leq 2 \E Z'.
  \end{equation}
  Furthermore, for each $t > 0$ we have the deviation inequality 
  \begin{equation}\label{sym2}
    P ( Z \geq 2 \E Z + t ) \leq 4P(Z' > t/2).
  \end{equation}
\end{lemma}

The proof of this result is identical to the proof of \cite[Lemma 5.1]{pv2012}.

The following is a standard Gaussian concentration inequality, which is a simple application of \cite[Theorem 7.1]{Ledoux}.

\begin{lemma}[Gaussian concentration]\label{gconcentration}
  Given a set $K \subseteq B_2^n$, we have
$$ P\parenth{ \sup_{\vz\in K - K} \ip{ \vg, \vz } - w(K) >  r } \leq e^{-r^2/8}, \ r > 0.$$ 
\end{lemma}

The following inequality is a specialization of \cite[Lemma 4.6]{LT}.  (In contrast to \cite[Lemma 4.6]{LT} we allow the consideration of a semi-norm, but the proof of the inequality remains unchanged.)

\begin{lemma}[Contraction Principle] \label{comparison}
Consider sequences of independent symmetric random variables $\eta_i$ and $\xi_i$ such that for some 
scalar $M \geq 1$, and every $i$ and $t > 0$,
$$P(|\eta_i| > t) \leq M P(|\xi_i| > t).$$
Let $\opnorm{\cdot}$ denote a semi-norm.  Then for any finite sequence $x_i$ and a scalar $p\geq 1$, we have
\[
\E \left(
\bigg\|\sum_{i=1}^{n} \eta_i x_i\bigg\|\right)^p \leq \E \left( M \bigg\|\sum_{i=1}^{n} \xi_i x_i\bigg\|
\right)^p.
\]
\end{lemma}

We will first apply Lemma~\ref{comparison} to derive a moment bound on $Z'$.  We then convert the moment bound back into a tail bound which we plug into the right-hand side of Equation \eqref{sym2}.

Because $\eps_i y_i \va_i$ has the same distribution as $\va_i$, and by the symmetry of $K - K$,
\[
 \E (Z')^p 
 = \E \parenth{ \sup_{\vz\in K - K} \frac{1}{m} \sum_{i = 1}^m \ip{ \va_i, \vz } }^p 
 = \E \parenth{ \sup_{\vz\in K - K} \frac{1}{m} \sum_{i = 1}^m \sum_{j = 1}^n (a_i)_j z_j }^p.
 \]
We apply Lemma~\ref{comparison} with $(a_i)_j$ in place of $\eta_i$, $e_i e_j^*$ in place of $x_i$ (where $e_i$ is the i-th standard basis vector), $\xi_i$ as independent $N(0,1)$ random variables, 
and the matrix semi-norm defined by $\| A \|: = \sup_{\vz\in K - K} \sum_{i,j} A_{i,j} z_j$. 
To this end, recall that $(a_i)_j$ are distributed identically with $a$. 
Since $a$ is a sub-gaussian random variable, it follows from definition \eqref{eq:subgaussian} that 
$$
P(|a| > t) \leq C P(|g|\cdot \kappa > t), \quad t > 0.
$$
Therefore an application of Lemma~\ref{comparison} allows us to replace
$(a_i)_j$ by $(C \kappa) (g_i)_j$ and thus conclude that 
\begin{equation}\label{pone}
  \E (Z')^p 
  \leq \E \parenth{\sup_{\vz\in K - K} 
  \frac{1}{m} \sum_{i = 1}^m \sum_{j = 1}^n C \k (g_i)_j z_j }^p 
  = \E \parenth{ \frac{C \kappa }{\sqrt{m}} \sup_{\vz\in K - K} \ip{ \vg, \vz } }^p.
\end{equation}

To further develop this inequality, we express the Gaussian concentration tail bound (Lemma~\ref{gconcentration})
in terms of moment bounds. For convenience, define 
$$
\xi = \sup_{\vz\in K - K} \ip{ \vg, \vz }.$$
Using Lemma~\ref{gconcentration} and the equivalence of sub-gaussian properties, for instance 
in \cite[Lemma 5.5]{Vershynin2012}, we have
$$
(\E (\xi - w(K))_+^p)^{1/p} \leq C\sqrt{p}.
$$
Above $(\xi - w(K))_+ := \max(\xi - w(K), 0)$.
Applying Minkowski's inequality gives
$$
(\E \xi^p)^{1/p} \leq (\E (\xi - w(K))_+^p)^{1/p} + (\E w(K)^p)^{1/p} \leq C\sqrt{p} + w(K).
$$
Combine this with Equation \eqref{pone} to give the moment bound
\begin{equation}
\label{eq:moment bound}
(\E (Z')^p)^{1/p} \leq C \cdot \frac{\k (\sqrt{p} + w(K))}{\sqrt{m}}.
\end{equation}
For convenience, set $\beta = \sqrt{p}$.  
We now use Markov's inquality to convert the moment bound to a tail bound.  Set $t =  e \cdot (\E (Z')^p)^{1/p}$ and note that by Equation \eqref{eq:moment bound} above,
\[t \leq C \cdot \frac{\k (\beta + w(K))}{\sqrt{m}}.\]
Further, by Markov's inequality we have
\begin{equation}
\label{eq:tail bound}
P(Z' \geq t) \leq \frac{\E(Z')^p}{t^p} \leq e^{-\beta^2}
\end{equation}

To complete the proof of the proposition, apply Lemma~\ref{symmetrization}:  The moment bound \eqref{eq:moment bound} with $p = 1$ controls $\E(Z')$ and the tail bound \eqref{eq:tail bound} controls the right-hand side of Equation \eqref{sym2}.

\section{Conclusion}
\par 
In contrast to standard compressed sensing, one-bit compressed sensing is infeasible when the measurement vectors are Bernoulli and the signal is extremely sparse.  Nevertheless, we show that when the signal is sparse, but not overly sparse, it may be recovered from Bernoulli (or more generally, sub-gaussian) one-bit measurements. To our knowledge, these are the first theoretical results in one-bit compressed sensing that specifically allow non-Gaussian measurements.

%

\newpage

\section*{Appendix}

\subsection{Proof of Lemma~\ref{fourth}}

We apply a Lindeberg replacement argument in a way similar to 
\cite[Proposition D.2]{tao2010random}. Define $v_j = (x_j, z_j)$, and let 
$\vg \in \R^n$ be a vector of independent standard Gaussian variables which is also independent of $\va$. Define
 $S_i = \sum_{j = 1}^{i - 1} a_j v_j + \sum_{j = i + 1}^n g_j v_j$
 and $\phi(v) = \theta(x)z$ (where $v = (x, z)$). Define $(S_{i})_{1}$ to be the $x$ component and $(S_{i})_{2}$ to be the $z$ component. 
  Then note by telescoping,
 $$|\E \theta(\ip{ \va, \vx })\ip{ \va, \vz } - \E \theta(\ip{ \vg, \vx })\ip{ \vg, \vz }| 
= \abs{\E \phi\parenth{\sum_{j = 1}^{n} a_j v_j} 
- \E \phi\parenth{\sum_{j = 1}^{n} g_j v_j} }$$
$$\leq \sum_{i = 1}^n |\E \phi(S_i + a_i v_i) - \E \phi(S_i + g_i v_i)|.$$

By Taylor's theorem with remainder, we have
\[
\phi(S_i + a_i v_i) = \phi(S_i) + \sum_{|\alpha| = 1} (a_iv_i)^\alpha 
\partial^\alpha \phi(S_i) 
+ \frac{1}{2} \sum_{|\alpha| = 2} (a_iv_i)^\alpha 
\partial^\alpha \phi(S_i)  + \frac{1}{6}\sum_{|\alpha| = 3} (a_iv_i)^\alpha 
\partial^\alpha \phi(S_i') 
\]
for some $S_i'$ on the line segment joining $S_i$ and $S_i + a_iv_i$. A similar result holds for 
$\phi(S_i + g_iv_i)$ with respective $S_i''$. Observe that since $\E a = \E g = 0$ and $\E a^2 = \E g^2 = 1$, 
the zeroth to second order terms cancel upon taking expectations in the difference.
$$
|\E \phi(S_i + a_iv_i) - \E \phi(S_i + g_iv_i)| = \frac{1}{6}\abs{\E\sum_{\abs{\alpha}=3}(a_{i}v_{i})^{\alpha}\partial^{\alpha}\phi(S_{i}')-\E\sum_{\abs{\alpha}=3}(g_{i}v_{i})^{\alpha}\partial^{\alpha}\phi(S_{i}'')}.
$$
Consider the first expectation on the right hand side. Observe that the partials in the error vanish except 
when at most one partial is taken on the second argument of $\phi$, yielding either $\theta''(x)$ 
or $\theta'''(x)z$. Furthermore, note that since $S_i' $ is on the line segment joining $S_i$ and $S_i + a_iv_i$, we may apply the bound $|(S_i')_2| \leq |(S_i)_2| + |a_i z_i|$ to conclude
$$\E \abs{\sum_{|\alpha| = 3} (a_iv_i)^\alpha \partial^\alpha \phi(S_i') } \leq C\E  |a_i|^3 (x_i^2 |z_i| + |x_i|^3) (\|\theta''\|_\infty + \|\theta'''\|_\infty (|(S_i)_2| + |a_iz_i|))$$
$$= C (x_i^2 |z_i| + |x_i|^3) (\tau_2 \E |a_i|^3 + \tau_3 (\E |(S_i)_2 a_i^3| + |z_i|\E a_i^4))$$

Observe that $(S_i)_2$ and $a_i$ are independent, and $(\E |(S_i)_2|)^2 \leq \E (S_i)_2^2 \leq 1$ by Cauchy-Schwarz and the fact that the variance of an independent sum is a sum of variances. Further observing that $|z_i| \leq 1$ and $\E |a|^3 \leq \E a^4$, we may collect terms to conclude
$$\E \abs{\sum_{|\alpha| = 3} (a_iv_i)^\alpha \partial^\alpha \phi(S_i') }  \leq C \|\vx\|_\infty (|x_i z_i| + x_i^2) (\tau_2 + 2\tau_3)\E a_i^4.$$

A similar bound follows for the remainder from the Gaussian expansion, and observe that the Gaussian remainder can be absorbed since $\E a^4 \geq \E a^2 = 1$.  Note that  
$$ \|\vx\|_\infty \left(\sum_{i = 1}^n |x_i z_i| + x_i^2 \right) \leq \|\vx\|_\infty (\| \vx\|_2 \| \vz\|_2 + \| \vx \|_2^2) \leq 2 \| x \|_\infty$$
so that summing over $i$ from 1 to $n$, 
\[
\bigg|\E \theta(\ip{ \va, \vx })\ip{\va, \vz } - \E \theta(\ip{ \vg, \vx })\ip{ \vg, \vz } \bigg|
\leq  C(\tau_2 + \tau_3)\E a^4 \| \vx \|_\infty,
\]
which completes the proof of the lemma.

\subsection{Total variation: sign function}\label{thm3p}

We consider the setting of Theorem \ref{thm2}, where $\theta(t) = (2 - p) \sign(t)$, with the additional assumption that  $\|a - g\|_{TV} \leq \eps$.

\begin{theorem}[Estimating a signal with no noise]\label{thm3}
We remain in the setting of Theorem~\ref{thm2} with the additional condition 
$\|a - g\|_{TV} \leq \eps$. Then for each $\beta > 0$, with probability at least 
$1 - 4e^{-\beta^2},$ 
the solution 
$\widehat{\vx}$ to the optimization problem~(\ref{optimization}) satisfies
\[
\|\vx - \widehat{\vx}\|_2^2 \leq C \sqrt{\k} \eps^{1/8} + \frac{C \k}{(2 - p)\sqrt{m}}(w(K) + \beta).
\]
\end{theorem}

To prove the theorem, we only need bound $\alpha$ in Lemma~\ref{lem:general structure}.  This is contained in the following proposition.  For simplicity, we will prove the theorem in the noiseless case when $p = 0$, but note that the noisy case follows from a simple rescaling argument as in the proof of Theorem \ref{thm2}. Below, $\lambda = \sqrt{2/\pi}$.

\begin{proposition}[Expectation]\label{expectation3}
For $\vx, \vx' \in B_2^n$,
$$|\E f_{\vx}(\vx') - \ip{ \lambda \vx, \vx' }|\leq C(\E a^4)^{1/8}\eps^{1/8}.$$
\end{proposition}
We now prove the theorem.

\begin{proof}[Proof of Theorem~\ref{thm3}]
By definition $\E a^4 \leq 16 \kappa^4$.  Thus, we may apply Lemma~\ref{lem:general structure} with 
\[\alpha = C \sqrt{\kappa} \eps^{1/8}\]
to complete the proof.
\end{proof}

To prove the proposition, we proceed with similar steps to the proof of Theorem~\ref{thm2}, starting with the following lemma.

\begin{lemma}\label{lem:variation-inner-product}
$|\ip{\vv_{\vx}, \vx} - \lambda | = \abs{\E \abs{ \< \va, \vx \>} - \sqrt{\frac{2}{\pi}} }\leq 4(\E a^4 + \E g^4)^{1/4}\eps^{1/4}$.
\end{lemma}

\begin{proof}
We first prove a variant of the Berry-Esseen result on expectations, applying Lindeberg replacement. Define $S_i = \sum_{j = 1}^{i - 1} a_j x_j + \sum_{j = i + 1}^n g_j x_j$, and $\phi(x)$ to be a bounded twice differentiable function. We will later use an approximation argument to replace $\phi$ by the absolute value function.  

Note by telescoping,
$$ \left|\E \phi(\ip{\va, \vx}) - \E \phi(\ip{\vg, \vx}) \right| = \left|\E\phi\left(\sum_{i = 1}^n a_i x_i\right) - \E\phi\left(\sum_{i = 1}^n g_i x_i\right)\right|$$
$$ \leq \sum_{i = 1}^n |\E\phi(S_i + a_i x_i) - \E\phi(S_i + g_i x_i)|.$$
For convenience, dropping subscripts, we now wish to bound $|\E\phi(S + ax) - \E\phi(S + gx)|$. 

By Taylor's theorem with remainder, we have
$$\phi(S + ax) = \phi(S) + ax\phi'(S) + R(S, ax)$$
where $|R(S, ax)| \leq (ax)^2 \|\phi''\|_\infty /2$. A similar result holds for $\phi(S + gx)$. 

Split $R(S, x)$ into $R_+(S, x) \geq 0$ and $R_-(S, x) \geq 0$. Observe that since $\E a = \E g = 0$, the zeroth and first order terms cancel upon taking expectations in the difference 
$$|\E\phi(S + ax) - \E\phi(S + gx)| = |\E R(S, ax) - \E R(S, gx))|$$
$$\leq |\E R_+(S, ax) - \E R_+(S, gx)| + |\E R_-(S, ax) - \E R_-(S, gx)|.$$

Consider the difference with $R_+$. We will apply the assumption $\|a - g\|_{TV} \leq \eps$. First, observe that $S$ is independent of both $a$ and $g$ and may be viewed as a constant. Viewing for instance $R_+(S, ax)$ as a function of $a$, 
$$\left|\int_0^M P(R_+(S, ax) > t) dt - \int_0^M P(R_+(S, gx) > t) dt\right| \leq M \eps.$$
Then, consider the tail of the first integral:
$$\int_M^\infty P(R_+(S, ax) > t) dt \leq \int_M^\infty \frac{\E(R_+(S, ax)^2)}{t^2} dt$$
$$= \frac{\E R_+(S, ax)^2}{M} \leq \frac{x^4 \E a^4 \|\phi''\|_\infty^2 }{4M}.$$
The Gaussian tail yields a similar error. Hence, optimizing over $M$ by choosing 
$$M = \frac{x^2(\E a^4 + \E g^4)^{1/2} \|\phi''\|_\infty}{2\sqrt{\eps}}$$ 
we have an overall error of
$$|\E R_+(S, ax) - \E R_+(S, gx))| \leq x^2(\E a^4 + \E g^4)^{1/2} \|\phi''\|_\infty\sqrt{\eps}.$$

The same holds for the difference with $R_-$. Finally, summing over the $n$ indices, and using that $\|x\|_2 = 1$,
$$ \left|\E \phi(\ip{\va, \vx}) - \E \phi(\ip{\vg, \vx}) \right| \leq 2(\E a^4 + \E g^4)^{1/2}\|\phi''\|_\infty \sqrt{\eps}.$$

Second, we approximate the absolute value using $\phi(x) := \sqrt{c + x^2} \approx |x|$. Observe for instance that $|\E|\ip{\va, \vx}| - \E\phi(\ip{\va, \vx})| \leq \sqrt{c}$, and likewise with $\vg$ in the place of $\va$. Evaluating $\phi''(x) = c/(c + x^2)^{3/2}$ with a maximum of $1/\sqrt{c}$ at $x = 0$, we may conclude
$$|\ip{\vv_x, \vx} - \lambda | = \bigg |\E|\ip{\va, \vx}| - \E|\ip{\vg, \vx}| \bigg| \leq 2\sqrt{c} + 2(\E a^4 + \E g^4)^{1/2}\sqrt{\frac{\eps}{c}}.$$
Choosing $\sqrt{c} = (\E a^4 + \E g^4)^{1/4}\eps^{1/4}$ completes the proof of the lemma.
\end{proof}

We now proceed to bound $\twonorm{\vv_x}$, thus obtaining the second geometric constraint required in the proof of the proposition.  
We apply Lemma~\ref{lem:variation-inner-product} with $\vz=\vv_x/\|\vv_x\|_2$ in the place of $\vx$:
\begin{align}
\|\vv_x\|_2 &= \langle \vv_{\vx}, \vv_{\vx}/\|\vv_{\vx}\|_2 \rangle = \E \sign(\langle \va, \vx\rangle) \langle \va, \vz \rangle \nonumber \\
&\leq \E \sign(\ip{ \va, \vz}) \ip{ \va, \vz}  = \ip{\vv_z, \vz}  \leq \lambda + 4(\E a^4 + \E g^4)^{1/4}\eps^{1/4}. \label{eq:bound vx}
\end{align}

We conclude the proof with the following calculation.
$$|\langle \vv_{\vx}, \vx' \rangle - \langle \lambda \vx, \vx' \rangle|^2 \leq \|\vv_{\vx}\|_2^2 - \lambda^2 + 2\lambda(\lambda - \langle \vv_{\vx}, \vx \rangle)$$
$$= (\|\vv_{\vx}\|_2 + \lambda)(\|\vv_{\vx}\|_2 - \lambda) + 2\lambda(\lambda - \langle \vv_{\vx}, \vx \rangle)$$
$$\leq 16(\E a^4 + \E g^4)^{1/4}\eps^{1/4}.$$

The last inequality follows from Equation \eqref{eq:bound vx} and Lemmas \ref{lem:variation-inner-product} and \ref{zeroth}.  Proposition~\ref{expectation3} is a consequence of absorbing constants.

\subsection{Total variation: smooth noise model}\label{thm4p}

We consider the setting of Theorem~\ref{thm1}, with the additional assumption that  $\|a - g\|_{TV} \leq \eps$. We also relax the assumption on $\theta(t)$, defined as in (\ref{noise}), to $\theta(t) \in C^2$.

\begin{theorem}[Estimating a signal with noise]\label{thm4}
We remain in the setting of Theorem~\ref{thm1} with the additional condition $\|a - g\|_{TV} \leq \eps$, and also relax the condition on $\theta(t)$ to $\theta(t) \in C^2$. Then for each $\beta > 0$, with probability at least 
$
1 - 4e^{-\beta^2},
$
the solution 
$\widehat{\vx}$ to the optimization problem~(\ref{optimization}) satisfies
\[
\|\vx - \widehat{\vx}\|_2^2 \leq C \left((\kappa^3 + 1)(\tau_1 + \tau_2)\sqrt{\eps} + \frac{\k}{\lambda\sqrt{m}}(w(K) + \beta) \right).
\]
\end{theorem}

We bound $\alpha$ in Lemma~\ref{lem:general structure}.

\begin{proposition}[Expectation]\label{expectation4}
For $\vx \in S^{n-1}, \vx' \in B_2^n$,
$$|\E f_{\vx}(\vx') - \ip{ \lambda \vx, \vx' }|\leq 8(\E a^6 + \E g^6)^{1/2} (\tau_1 + \tau_2)\sqrt{\eps}.$$
\end{proposition}

We now prove the theorem.

\begin{proof}[Proof of Theorem~\ref{thm4}]
By definition of $\kappa$ we have $\E a^6 \leq 216 \kappa^6$ which implies $\sqrt{\E a^6 + \E g^ 6 } \leq C(\kappa^3 + 1)$.  Thus, we apply Lemma~\ref{lem:general structure} with 
\[\alpha = C(\kappa^3 + 1)(\tau_1 + \tau_2)\sqrt{\eps}\]
to complete the proof of the theorem.
\end{proof}

We now prove the proposition.

\begin{proof}[Proof of Proposition~\ref{expectation4}]
Recalling the steps in Section \ref{thm1p}, observe that the left hand side of the inequality is expressible as
$$|\E f_{\vx}(\vx') - \ip{ \lambda \vx, \vx' }| = |\E \theta(\ip{\va, \vx}) \ip{\va, \vx'} - \E \theta(\ip{\vg, \vx}) \ip{\vg, \vx'} |.$$

The statement of the proposition becomes similar to that of Lemma~\ref{fourth}. Using the same notation and proceeding as in its proof (including the use of $\vz$ in place of $\vx'$), we apply Lindeberg replacement:
$$|\E \theta(\ip{ \va, \vx })\ip{ \va, \vz } - \E \theta(\ip{ \vg, \vx })\ip{ \vg, \vz }| \leq \sum_{i = 1}^n |\E \phi(S_i + a_i v_i) - \E \phi(S_i + a_i v_i)|.$$
As before, we Taylor expand, except only to second order error: 
\[
\phi(S_i + a_i v_i) = \phi(S_i) + \sum_{|\alpha| = 1} (a_iv_i)^\alpha \partial^\alpha \phi(S_i) + R(S_i, a_i v_i)
\]
where $R(S_i, a_i v_i) = \frac{1}{2} \sum_{|\alpha| = 2} (a_iv_i)^\alpha 
\partial^\alpha \phi(S_i') $ for some $S_i'$ on the line segment joining $S_i$ and $S_i + a_iv_i$. A similar result holds with 
$\phi(S_i + g_iv_i)$, with respective $S_i''$. 

Split $R(S, v)$ into $R_+(S, v) \geq 0$ and $R_-(S, v) \geq 0$. Observe that since $\E a = \E g = 0$, the zeroth and first order terms cancel upon taking expectations in the difference 
$$|\E\phi(S_i + a_iv_i) - \E\phi(S_i + g_iv_i)| = |\E R(S_i, a_i v_i) - \E R(S_i, g_i v_i))|$$
$$\leq |\E R_+(S_i, a_i v_i) - \E R_+(S_i,  g_i v_i)| + |\E R_-(S_i,  a_i v_i) - \E R_-(S_i, g_i v_i)|.$$

Consider the difference containing $R_+$. We will apply the assumption $\|a - g\|_{TV} \leq \eps$. First, observe that $S_i$ is independent of both $a_i$ and $g_i$ and may be viewed as a constant (by conditioning on it). Viewing for instance $R_+(S_i, a_i v_i)$ as a function of $a_i$,
$$\left|\int_0^M P(R_+(S_i, a_i v_i) > t) dt - \int_0^M P(R_+(S_i, g_i v_i) > t) dt\right| \leq M \eps.$$
Then, consider the tail of the first integral:
$$\int_M^\infty P(R_+(S_i, a_i v_i) > t) dt \leq \int_M^\infty \frac{\E(R_+(S_i, a_i v_i)^2)}{t^2} dt = \frac{\E R_+(S_i, a_i v_i)^2}{M}.$$

Recall the explicit form of the remainder and observe that the partials in the error vanish except 
when at most one partial is taken on the second argument of $\phi$, yielding either $\theta'(x)$ 
or $\theta''(x)z$. Furthermore, note that since $S_i' $ is on the line segment joining $S_i$ and $S_i + a_iv_i$, we may apply the bound $|(S_i')_2| \leq |(S_i)_2| + |a_i z_i|$ to conclude
$$\E 4R_+(S_i, a_i v_i)^2 \leq \E 4R(S_i, a_i v_i)^2  \leq \E \left( \sum_{|\alpha| = 2} a_i^2 |v_i^\alpha| (\|\theta'\|_\infty + \|\theta''\|_\infty (|(S_i)_2| + |a_iz_i|)) \right)^2$$
$$= \E \left(a_i^2 (|x_i| + |z_i|)^2 (\tau_1 + \tau_2 (|(S_i)_2| + |z_ia_i|) )\right)^2$$

Observe that $(S_i)_2$ and $a_i$ are independent, and $(\E |(S_i)_2|)^2 \leq \E (S_i)_2^2 \leq 1$ by Cauchy-Schwarz and that the variance of an independent sum is a sum of variances. Further observing that $|z_i| \leq 1$ and for instance $\E |a|^5\leq \E a^6$, rearranging and collecting terms yields
$$\E 4R_+(S_i, a_i v_i)^2 \leq (|x_i| + |z_i|)^4  \left( 4\tau_2^2 \E a_i^6 + \tau_1^2 \E a_i^4+ 4\tau_1\tau_2\E |a_i|^5 \right) $$
$$\leq 4(|x_i| + |z_i|)^4  (\tau_1 + \tau_2)^2 \E a_i^6.$$
The Gaussian tail yields a similar error. Hence, optimizing over $M$ by choosing 
$$M = \frac{1}{\sqrt{\eps}}(|x_i| + |z_i|)^2(\E a^6 + \E g^6)^{1/2} (\tau_1 + \tau_2)$$ 
we have overall error
$$|\E R_+(S_i, a_i v_i) - \E R_+(S_i, g_i v_i)| \leq 2(|x_i| + |z_i|)^2(\E a^6 + \E g^6)^{1/2} (\tau_1 + \tau_2)\sqrt{\eps}.$$

The same holds for the difference with $R_-$. Finally, summing over the $n$ indices, and using that $\|\vx\|_2 = 1$ and $\|\vz\|_2 = 1$,
$$ \left|\E \phi(\ip{\va, \vx}) - \E \phi(\ip{\vg, \vx}) \right| \leq 8(\E a^6 + \E g^6)^{1/2} (\tau_1 + \tau_2)\sqrt{\eps},$$
which concludes the proof of the proposition.
\end{proof}

\bibliographystyle{acm}
\bibliography{pv-1bitcs-robust}

%
%
%
%
%
%
%
%
%

\end{document}